\documentclass[a4paper,fleqn]{cas-sc}

\usepackage[numbers]{natbib}
\usepackage{amssymb}
\usepackage{mathrsfs}
\usepackage{amsmath}
\usepackage{accents}
\usepackage{amsthm}
\usepackage{empheq}
\usepackage[subnum]{cases}
\usepackage{cancel}
\usepackage{wrapfig}
\usepackage{tikz-cd}
\usetikzlibrary{positioning}
\usepackage{subeqnarray}

\newtheorem{corol}{Corollary}
\newtheorem{theorem}{Theorem}
\newtheorem{defn}{Definition}
\newtheorem{rem}{Remark}
\newtheorem{Note}{Note}

\newcommand*{\QEDA}{\hfill\ensuremath{\blacksquare}}

\definecolor{myblue}{rgb}{.8, .8, 1}

\newlength{\dhatheight}
\newcommand{\doublehat}[1]{%
    \settoheight{\dhatheight}{\ensuremath{\widehat{#1}}}%
    \addtolength{\dhatheight}{-0.35ex}%
    \widehat{\vphantom{\rule{1pt}{\dhatheight}}%
    \smash{\widehat{#1}}}}
\def\tsc#1{\csdef{#1}{\textsc{\lowercase{#1}}\xspace}}
\tsc{WGM}
\tsc{QE}
\tsc{EP}
\tsc{PMS}
\tsc{BEC}
\tsc{DE}

\begin{document}
\let\WriteBookmarks\relax
\def\floatpagepagefraction{1}
\def\textpagefraction{.001}
\shorttitle{Dark Fields do Exist in Weyl Geometry}
\shortauthors{}

\title [mode = title]{Dark Fields do Exist in Weyl Geometry}                      


\author[1]{Fereidoun Sabetghadam}[type=editor,
                        auid=,bioid=,
                        prefix=,
                        role=,
                        orcid=0000-0001-8412-8324]

\ead{fsabet@srbiau.ac.ir}











\begin{abstract}
A generalized Weyl integrable geometry (GWIG) is obtained from simultaneous affine transformations of the tangent and cotangent bundles of a (pseudo)-Riemannian manifold. In comparison with the classical Weyl integrable geometry (CWIG), there are two generalizations here: interactions with an arbitrary dark field, and, anisotropic dilation. It means that CWIG already has interactions with a {\it null} dark field. Some classical mathematics and physics problems may be addressed in GWIG. For example, by derivation of Maxwell's equations and its sub-sets, the conservation, hyperbolic, and elliptic equations on GWIG; we imposed interactions with arbitrary dark fields. Moreover, by using a notion analogous to Penrose conformal infinity, one can impose boundary conditions canonically on these equations. As a prime example, we did it for the elliptic equation, where we obtained a singularity-free potential theory. Then we used this potential theory in the construction of a non-singular model for a point charged particle. It solves the difficulty of infinite energy of the classical vacuum state. 
\end{abstract}




\maketitle

\section{Introduction} \label{intro}
Given a (psudo)-Riemannian ${\mathsf d}$-manifold $(M,g)$, a Weyl integrable manifold is defined by $(M,[g],\lambda)$, where $[g]$ is the conformal class in which ${g}^\prime \sim {g}$ iff ${g}^\prime=e^{2\lambda(x)}{g}$ with smooth $\lambda:M\rightarrow {\Bbb R}$ where $x\in M$ \cite{Folland}. Our manifolds are finite-dimensional ${\mathsf d}<\infty$, and we study the classical fields. The quantum and relativistic fields may be considered as the extensions.

The idea of relating Weyl geometry with the dark fields has a long history \cite{Scholz2019}. Half a century ago, Dirac \cite{Dirac1973} in seeking gauge-invariant general relativity modified Weyl geometry. The new geometry interpreted later as the necessity of the presence of a dark matter \cite{Israelit1992}. In the next years, more or less similar ideas have appeared from time to time in different forms in the works of different researchers \cite{Cheng1988, Maeder2020}.  

Here we obtain a generalized Weyl integrable geometry (GWIG) by affine transformations of the pair ``{\it vector space-- dual space}'' at each point $x\in M$. An arbitrary pair ``(vector field, 1-form)'' appears in these transformations which can be interpreted as dark fields. A direct result of the approach is that the classical Weyl integrable geometry (CWIG) has already included interactions with a {\it null} dark field. The author first met a nave version of the transformations in imposition of the no-slip conditions on the Navier-Stokes equations \cite{Sabetghadam2015} (without aware of Weyl geometry). However, soon Weyl geometry presented itself. 

In the sequel, we first provide our suggested transformation in $\S$ \ref{Affin_trans}, which is used in obtaining GWIG and its properties in $\S$ \ref{Revised_IWG}. Our physical applications are provided in sections \ref{electromagnetism}--\ref{non-sing} and appendix \ref{append1}, where we obtained Maxwell, conservation, hyperbolic and elliptic equations on GWIG, and some theorems for their general solutions. Finally we provide a summary in $\S$ \ref{conclusion}.    
\section{A tensorial affine transformation} \label{Affin_trans}
We transform both tangent and cotangent bundles of a (pseudo)-Riemannian manifold simultaneously. The transformations are affine; and the affinity parameters are tensorial (not scalars), to admit anisotropy. 
\begin{defn}
Consider a (pseudo)-Riemannian $\mathsf d$-manifold $(M,g)$ and its tangent and cotangent fiber bundles $(TM \xrightarrow{\pi_v} M)$ and $(T^*M\xrightarrow{\pi_{\alpha}} M)$, containing the tangent vectors $v\in\Gamma(TM)$ and the differential 1-forms $\alpha\in\Omega^1(M)$.\\
We define the bundlemorphisms $\pmb{\mathfrak F}:=(\pmb{f}_{\pmb{\mathsf w}^v,v_d},\pmb{f}_{\pmb{\mathsf w}^\alpha,\alpha^d}):(\widehat{T}M,\widehat{T}^*M)\hookrightarrow(TM,T^*M)$ that $(v,\alpha)\mapsto(\widehat{v},\widehat{\alpha})$ as: 
\begin{numcases}{}
\widehat{v}:=\pmb{f}_{\pmb{\mathsf w}^v,v_d}(v)\,:=\big({\bf 1}-{\pmb{\mathcal K}^v}\big)(v)~\;+{\pmb{\mathcal K}^v}(v_d),
 \label{f} \\
\widehat{\alpha}:=\pmb{f}_{\pmb{\mathsf w}^\alpha,\alpha^d}(\alpha):=\big({\bf 1}-{\pmb{\mathcal K}^\alpha}\big)(\alpha)+{\pmb{\mathcal K}^\alpha}(\alpha^d),
\label{f*}
\end{numcases}
where $\pmb{\mathcal K}^v\in{\rm Aut}(TM)$ and $\pmb{\mathcal K}^\alpha\in{\rm Aut}(T^*M)$ are defined as:
\begin{numcases}{}
{\pmb{\mathcal K}^v}(\kappa):={\rm diag}\big({\mathcal K}^v_j\big):={\rm diag}\Big(1-\big(1-\kappa\big)^{{\mathsf w}^v_j}\Big),
 \label{K}\\
{\pmb{\mathcal K}^\alpha}(\kappa):={\rm diag}\big({\mathcal K}^\alpha_j\big):={\rm diag}\Big(1-\big(1-\kappa\big)^{{\mathsf w}^\alpha_j}\Big), \label{K*}
\end{numcases}
for the suitable choices of  $v_d\in \Gamma(TM)$, $\alpha^d\in\Omega^1(M)$, $\pmb{\mathsf w}^v:=\{{\mathsf w}^v_j\}_{j=0}^{\mathsf d-1}$ and $\pmb{\mathsf w}^\alpha:=\{{\mathsf w}^\alpha_j\}_{j=0}^{\mathsf d-1}$ where ${\mathsf w}^v_j,{\mathsf w}^\alpha_j\in {\Bbb R}$ are anisotropic Weyl weights; and $\kappa: M\to \ell$ is a free scalar affinity parameter, that is, $\ell:=[0,1)$.
\end{defn}
\noindent
The usefulness of transformations (\ref{f})--(\ref{K*}) is manifested via the following statement and note:
\begin{theorem}\label{the00}
For $\alpha=g(v)$, we have $\widehat{\alpha}=h(\widehat{v})$ where:
\begin{equation}\label{metric00}
h:=\left(\pmb{f}_{\pmb{\mathsf w}^\alpha,\alpha^d}\right)\circ g \circ \left(\pmb{f}_{\pmb{\mathsf w}^v,v_d}\right)^{-1}.
\end{equation}
\end{theorem}
\begin{proof}
The above diagram should commute.
    \begin{tikzpicture}[node distance = 1.5cm, thick, black]%
        \node (1) {$v$};
        \node (2) [right=of 1] {$\widehat{v}$};
        \node (3) [below=of 2] {$\widehat{\alpha}$};
        \node (4) [left=of 3]  {$\alpha$};
        \draw[->] (1) -- node [midway,above] {${f}_{\pmb{\mathsf w}^v,v_d}$} (2);
        \draw[->] (1) -- node [midway,left]{$g$} (4);
        \draw[->] (2) -- node [midway,right, red] {$h=({f}_{\pmb{\mathsf w}^{\alpha},\alpha^d})\circ g \circ ({f}_{\pmb{\mathsf w}^v,v_d})^{-1}$} (3);
        \draw[->] (4) -- node [midway,below] {${f}_{\pmb{\mathsf w}^\alpha,\alpha^d}$} (3);
    \end{tikzpicture}%
\end{proof}
\begin{Note}\label{note2_n}
We emphasize that: $\left(\pmb{\mathcal K}^v(0),\pmb{\mathcal K}^\alpha(0)\right)=({\bf 0},{\bf 0})$ and $\left(\pmb{\mathcal K}^v(1),\pmb{\mathcal K}^\alpha(1)\right)=({\bf 1},{\bf 1})$. It results in the following interesting properties for the bundlemorphisms $\pmb{\mathfrak{F}}$:
\begin{enumerate}
\item Everywhere that $\kappa(x)=0$, we have $\pmb{\mathfrak{F}}=\left({{\rm id}_{{\rm Aut}(TM)}},{\rm id}_{{\rm Aut}(T^*M)}\right)$, that is, $(\widehat{v},\widehat{\alpha})=(v,\alpha)$, which means the genuine Riemannian geometry. 
\item Everywhere that $\kappa(x)\to 1$, we have $\lim_{\kappa\to 1}\left(\pi_v^{-1}(x),\pi_\alpha^{-1}(x)\right)=(v_d,\alpha^d)_x$, that is, the vector spaces $(\widehat{T}_xM,\widehat{T}^*_xM)$ shrink to the (vector, covector) $\left((v_d)_x,(\alpha^d)_x\right)$. \label{item2}
\item Meanwhile, we have $\widehat{\alpha}=h(\widehat{v})$ for any $0\leq\kappa<1$.
\end{enumerate}
\end{Note}
\begin{rem}
 On the contrary to item \ref{item2} above, one can suppose that some primary vector filed and 1-form $(v_d,\alpha^d)$ have been  expanded to the bundles $(\widehat{T}M,\widehat{T}^*M)$. Based on this interpretation, we suggest the name ``primary fields'' instead of the ``dark fields'' for $(v_d,\alpha^d)$. 
\end{rem}
On one spacetime manifold, different physical quantities may have different $(\pmb{\mathsf w}^v,v_d)$ and $(\pmb{\mathsf w}^\alpha,\alpha^d)$, and therefore, different $\pmb{\mathfrak{F}}$. Here we distinguish the displacements on the spacetime manifold from the other tangent vector fields. 
\section{A generalized Weyl integrable geometry (GWIG)} \label{Revised_IWG}
We distinguish the displacements on the spacetime manifold from the other tangent vector fields; and for each one we define a sub-set of transformations (\ref{f})-(\ref{K*}). Then, together, they form the generalized Weyl integrable geometry that includes interaction with the dark (primary) fields. These are provided in the following two sub-sections.
\subsection{Displacement on the spacetime manifold}\label{spacetimea-man}
For the displacement vectors (on the spacetime manifold), we suggest $\alpha^d=v_d=0$, which is consistent with the classical picture of the space-time manifold. Moreover, we define ${\mathsf w}^\alpha_j=:{\mathsf z}_j$, and we assume ${\mathsf w}^v_j=0$.  Therefore:
\begin{equation}\label{eq15}
\pmb{\mathfrak{F}}_X:=\left(\pmb{f}_{\pmb{0},0},\pmb{f}_{\pmb{\mathsf z},0}\right) =(\pmb{1},\pmb{\chi}),
\end{equation}
where $\pmb{\mathsf z}:=\{{\mathsf z}_j\}_{j=0}^{\mathsf d-1}$, $\mathsf{z}_j\in {\Bbb R}$; and  
\begin{equation}
\pmb{\chi}:={\rm diag}(\chi_j):={\rm diag}\Big((1-\kappa)^{{\mathsf z}_j}\Big). \label{chi}
\end{equation}
We emphasize that $\pmb{\chi}(0)={\bf 1}$ and $\pmb{\chi}(1)={\bf 0}$.
\begin{rem}\label{rem_2}
By the term `dilation', one usually means dilation in spacetime. Therefore, we consider $\pmb{\mathsf z}$ as the set of parameters that defines the dilations, and we call  $\pmb{\chi}$ as the dilation tensor. Then, other functions on $M$ may have their own $\left(v_d,\pmb{\mathsf w}^v(\pmb{\mathsf z})\right)$ and $\left(\alpha^d,\pmb{\mathsf w}^\alpha(\pmb{\mathsf z})\right)$. Note that anisotropic dilations are permitted because ${\mathsf z}_j$ may be different. Moreover, since $\pmb{\chi}$ is diagonal, we will use $\chi_j=\chi^j$ wherever it is needed in Einstein summation convention.
\end{rem}
\noindent
The relation between the displacement vector fields and their associated differential 1-forms is obtainable directly from theorem \ref{the00}:
\begin{corol}\label{coroll_1}
Let $\xi\in TM$ be a displacement vector field on the manifold $(M,g)$, and $(\xi)^\flat=:\eta\in T^*M$. Then, under the transformation $\pmb{\mathfrak{F}}_X$, we have $\widehat{\eta}=\widehat{g}(\xi)$, where $\widehat{g}=\pmb{\chi}\circ g$. 
\end{corol}
\begin{proof}
The above diagram should commute.
    \begin{tikzpicture}[node distance = 1.5cm, thick, black]%
        \node (1) {$\xi$};
        \node (2) [below=of 1] {$\eta$};
        \node (3) [right=of 2]  {$\widehat{\eta}$};
        \draw[->] (1) -- node [midway,left] {$g$} (2);
        \draw[->] (2) -- node [midway,below]{$\pmb{\chi}$} (3);
        \draw[->] (1) -- node [midway,right, red] {$~~\widehat{g}=\pmb{\chi}\circ g$} (3);
    \end{tikzpicture}%
\end{proof}
\begin{Note}
Unless for the isotropic dilation, $\widehat{g}$ is not symmetric. Moreover, for $\mathsf{z}_j > 0 $ by $\kappa=0$, we have $\widehat{g}=g$, and $\lim_{\kappa \to 1}\widehat{g}=0$.
\end{Note}
Now, let $\{\partial_j,dx^j\}_{j=0}^{\mathsf d-1}$ be the {\it standard} local frames on $(TM,T^*M)$, where $dx^i(\partial_j)=\delta^i_{\;j}$. Then, under $\pmb{\mathfrak{F}}_X$ they transform as:
\begin{numcases}{}
\widehat{\partial}_j=\partial_j, \label{partial_1}\\
\widehat{dx}^j={\pmb{\chi}}(dx^j)=\chi^j\;dx^j=(1-\kappa)^{\mathsf z_j}\;dx^j. \label{dual_f_1}
\end{numcases}
That is, merely the dual frame is affected while the coordinate vectors are remained unchanged. The following definition and theorem relate this geometry with the classical Weyl integrable geometry:
\begin{defn}\label{Def3}
We distinguish a Riemannian observer $\mathscr{R}:=\{\partial_j,\widehat{dx}^j\}_{j=0}^{\mathsf{d}-1}$ from a Weylian observer $\mathscr{W}:=\{\partial_j,dx^j\}_{j=0}^{\mathsf{d}-1}$.
\end{defn}
\begin{Note}
Choosing $\{\widehat{dx}^j\}_{j=0}^{\mathsf{d}-1}$ as the dual basis is not a part of transformation $\pmb{\mathfrak{F}}_X$. It is just choosing a new basis on a transformed vector space, equivalent to choosing $\widehat{dx}^i(\partial_{j})={\rm diag}(\chi^i)$, instead of the conventional choice $dx^i(\partial_{j})=\delta^i_{\;j}$.
\end{Note}
\begin{defn}\label{Def4}
The representations of $\widehat{g}$ on $\mathscr R$ and $\mathscr W$ are denoted by $(\widehat{g}_{\mathscr R})_{ij}$ and $(\widehat{g}_{\mathscr W})_{ij}$, that is: 
\begin{equation}
\widehat{g}=:(\widehat{g}_{\mathscr R})_{ij}\; \widehat{dx}^i\otimes \widehat{dx}^j=:(\widehat{g}_{\mathscr W})_{ij}\;dx^i\otimes dx^j.
\end{equation}
\end{defn}
\begin{corol}
One can write: 
\begin{equation}
(\widehat{g}_{\mathscr W})_{ij}=\chi_j\; g_{ij},
\end{equation}
and: 
\begin{equation}
(\widehat{g}_{\mathscr R})_{ij}=e^{{\Lambda}_i}(\widehat{g}_{\mathscr W})_{ij}e^{{\Lambda}_j},
\end{equation}
in which we defined ${\Lambda}:={\rm diag}({\mathsf z}_j\lambda)$, such that $e^{-{\Lambda}}:=\chi$, where $\chi$ is the matrix representation of $\pmb \chi$.
\end{corol}
\begin{proof}
The proofs are immediate.
\end{proof}
\noindent
Note that by the above definitions $\kappa(x)$ is related to Weyl dilation function $\lambda(x)$ by:
\begin{equation}\label{kap_lamb_rel}
e^{-\lambda(x)}:=1-\kappa(x),
\end{equation}
which means for $\kappa\in [0,1)$, we have $\lambda\in [0,+\infty)$. Manifestly, in $\widehat{g}=\pmb{\chi}\circ g$, by ${\mathsf z}_j=2$ the classical Weyl metric ${g}_{\mathscr W}=e^{-2\lambda(x)}g_{\mathscr R}$ is retrieved. 
\begin{rem}\label{single_metric}
As a summary, $\pmb{\mathfrak{F}}_X$ associates $\widehat{g}=\pmb{\chi}(g)$ to $M$. Then, $\widehat{g}$ has two representations on ${\mathscr R}$ and ${\mathscr W}$ as $(\widehat{g}_{\mathscr W})_{ij}=\chi_j\;g_{ij}$, and $(\widehat{g}_{\mathscr R})_{ij}=e^{{\Lambda}_i}(\widehat{g}_{\mathscr W})_{ij}e^{{\Lambda}_j}=\chi_i^{-1}\;g_{ij}$. 
\end{rem} 
\subsection{Other tangent vector fields}
For the tangent vector fields, other than displacements on the spacetime, we choose another su-bset of (\ref{f})-(\ref{K*}) as follows. Manifestly, $(v_d,\alpha^d)$ are the fixed points of $(\pmb{f}_{\pmb{\mathsf w}^v,v_d},\pmb{f}_{\pmb{\mathsf w}^\alpha,\alpha^d})$. In the most general form, they are related to each other by $\alpha^d=\pmb{T}(g(v_d))$, where $\pmb{T}\in{\rm Aut}(T^*M)$. In the present work we merely study the classical fields. Therefore, we  restrict ourselves to the case that $\pmb{T}={\rm id}_{T^*M}$. Moreover, we need only one set of Weyl weights $\pmb{\mathsf w}$. Summarily:
\begin{numcases}{}
\alpha^d=g(v_d)=(v_d)^\flat, \label{alph_v} \\
\pmb{\mathsf w}^v=\pmb{\mathsf w}^\alpha=\pmb{\mathsf w}(\pmb{\mathsf z}).  \label{wj}
\end{numcases}
Therefore, the transformation for a general tangent vector field (and its dual) reads:
\begin{equation}
\pmb{\mathfrak{F}}_V:=\left(\pmb{f}_{\pmb{\mathsf w},v_d},\pmb{f}_{\pmb{\mathsf w},(v_d)^\flat}\right).
\end{equation}
Then, by writing the relation between a transformed vector $\widehat{v}$ and its associated 1-form $\widehat{\alpha}$ as $\widehat{\alpha}=h(\widehat{v})$, theorem \ref{the00} implies:
\begin{equation}\label{g-g}
h=\left(\pmb{f}_{\pmb{\mathsf w},(v_d)^\flat}\right)\circ g \circ \left(\pmb{f}_{\pmb{\mathsf w},v_d}\right)^{-1}.
\end{equation}

\noindent 
There are two important cases that $h=g$:
\begin{corol}\label{corol1}
In the following two cases, for $v=(\alpha)^\sharp$, we have $\widehat{v}=(\widehat{\alpha})^\sharp$, that is, $h=g$:
\begin{enumerate}
\item ${\mathsf w}_j$ s are the same, say ${\mathsf w}_j=\mathsf{w}$;
\item $g$ is diagonal.
\end{enumerate}
\end{corol}
\begin{proof}
They can be checked directly from Eq. (\ref{g-g}). 
\end{proof}
\subsection{A summary}
The bundle morphism $\pmb{\mathfrak F}_X$ acts on the tangent and cotangent bundles of a spacetime (pseudo)-Riemannian ${\mathsf d}$-manifold $(M,g)$. It affects the dual frame bundle, and leaves the frame bundle unchanged. Accordingly, two observers $\mathscr R$ and $\mathscr W$ may be defined that differ in their dual frames. Moreover, $\pmb{\mathfrak F}_X$ affects the other fields on $M$; the affections which is modeled by another bundlemorphism $\pmb{\mathfrak F}_V$. Therefore, the pair $(\pmb{\mathfrak F}_X,\pmb{\mathfrak F}_V)$ defines GWIG. It may be observed by $\mathscr R$ or $\mathscr W$. 
\subsection{$\mathscr R$ and $\mathscr W$ observations}\label{R-W}
$\mathscr R$ and $\mathscr W$ have different dual frames. Therefore, they observe functions differently. Here we summarize the differences, that are especially evident in the extreme situations $\kappa=0$ and $\kappa\to 1$.  
\begin{enumerate}
\item For $\kappa=0$, we have $(\widehat{g}_{\mathscr R})_{ij}=(\widehat{g}_{\mathscr W})_{ij}=g_{ij}$; and $\lim_{\kappa\to 1}(\widehat{g}_{\mathscr R})_{ij}=\infty$ while $\lim_{\kappa\to 1}(\widehat{g}_{\mathscr W})_{ij}=0$. 
\item An arbitrary transformed 1-form $\widehat{\alpha}$ has the representations $\widehat{\alpha}_j$ and $\doublehat{\alpha}_j$ on $\mathscr R$ and $\mathscr W$ respectively, that is: 
\begin{equation}
\widehat{\alpha}=\widehat{\alpha}_j\;\widehat{dx}^j=\doublehat{\alpha}_j\;{dx}^j.
\end{equation}
It means $\doublehat{\alpha}_j=\chi_j\;\widehat{\alpha}_j$. On the other hand, $\pmb{f}_{\pmb{\mathsf w},\alpha^d}$ yields $\widehat{\alpha}_j=(1-\mathcal{K}_j)\alpha_j+\mathcal{K}_j\alpha_j^d$.\\ 
Now, an interesting property is that while $\lim_{\kappa\to 1}\widehat{\alpha}_j=\alpha_j^d$, we have $\lim_{\kappa\to 1}\doublehat{\alpha}_j=0$, that is:
\begin{quote}
``Everywhere that $\kappa\to 1$, what that measuers as $\alpha^d$ by $\mathscr R$, it measuers as zero by $\mathscr W$.''
\end{quote}
Based on this property, we will suggest a notion, analogous to Penrose conformal infinity, which enables us to impose canonically the  boundary conditions on the solutions of Maxwell's equations.
\item The components of the vector fields $\widehat{v}_j$are the same for both $\mathscr R$ and $\mathscr W$, because their frame bundles are the same. Consequently, the scalar products have different results on $\mathscr R$ and $\mathscr W$, that is, $\mathfrak{r}_{\mathscr R}:=\widehat{\alpha}(\widehat{v})=\widehat{\alpha}_j\widehat{v}^j$, while $\mathfrak{r}_{\mathscr W}:=\widehat{\alpha}(\widehat{v})=\doublehat{\alpha}_j\widehat{v}^j=\chi_j\widehat{\alpha}_j\widehat{v}^j$. Now, as one can see: $\lim_{\kappa\to 1}\mathfrak{r}_{\mathscr R}= \alpha^d_iv_d^i$, while $\lim_{\kappa\to 1}\mathfrak{r}_{\mathscr W}=0$.
\item \label{item4-3} There is an arbitrariness in definition of the volume forms on a non-Riemannian manifold \cite{Eisenhart1926}. The genuine Riemannian volume form $\mathcal{V}=\sqrt{|{\rm det}(g)|} d{\mathscr{V}}$ is an invariant of $\pmb{\mathfrak F}_X$. On the other hand, by definition of the volume form $\widehat{\mathcal{V}}=:\widehat{\mathcal{V}}_{\mathscr R}\;\widehat{d\mathscr{V}}$, where $\widehat{\mathcal{V}}_{\mathscr R}=\sigma \sqrt{|\rm{det}(g)|}$ and $\sigma:=\sqrt{|{\rm det}(\chi^{-1})|}=(1-\kappa)^{-\frac{1}{2}\Sigma{\mathsf z}_j}$ is the dilation density; we have $\lim_{\kappa\to 1}\widehat{\mathcal{V}}=0$, while $\lim_{\kappa\to 1}\widehat{\mathcal{V}}_{\mathscr R}=\infty$. In other words, by the above definitions:
\begin{quote}
``Action of $\pmb{\mathfrak F}_X$ decreases the volume $\widehat{\mathcal V}$ and increases the density $\widehat{\mathcal{V}}_{\mathscr R}$, such that the initial Riemannian volume ${\mathcal V}=\widehat{\mathcal V}|_{\kappa=0}$ remains constant.''
\end{quote}

\end{enumerate}
\noindent
We are ready to study Maxwell's equations on this GWIG.  
\section{Maxwell's equations in the presence of dark (primary) fields} \label{electromagnetism}
On ${\Bbb R}^{(1,3)}$, consider the source-free Maxwell's equations:
\begin{numcases}{}
d{\bf F}=0, \label{Maxwell_1} \\
\delta{\bf F}=0,\label{Maxwell_2}
\end{numcases}
where $\bf F$ is the Faraday's 2-form. Eq. (\ref{Maxwell_1}) implies ${\bf F}=dA$ where $A$ is a gauge field.

Now, we assume that there is a dilation distribution $\pmb{\chi}$ and a dark (primary) electromagnetic 4-potential $A^d=(v_d)^\flat$ on ${\Bbb R}^{(1,3)}$, such that $(\pmb{\mathfrak F}_X,\pmb{\mathfrak F}_V)$ acts on its tangent and cotangent bundles. Therefore, $A$ can be written as:
\begin{equation}\label{AR}
A=\pmb{f}_{\pmb{\mathsf w},\alpha^d}^{-1}(\widehat{A})= \Big({\bf 1}-{\pmb{\mathcal K}^{\alpha}}\Big)^{-1}\big(\widehat{A}-{\pmb{\mathcal K}^{\alpha}}(A^d)\big),
\end{equation}
resulting in the gauge transformed Faraday's 2-form:
\begin{equation}
 {\bf F}={d}\Big(\big({\bf 1}-{\pmb{\mathcal K}^{\alpha}}\big)^{-1}\big(\widehat{A}-{\pmb{\mathcal K}^{\alpha}}(A^d)\big)\Big)=:{\bf \widehat{F}}-{\bf F}^d, \label{ext} 
\end{equation}
where:
\begin{equation}\label{new_Fs}
{\bf \widehat{F}}:={d}\Big(\big({\bf 1}-{\pmb{\mathcal K}^{\alpha}}\big)^{-1}\big(\widehat{A}\big)\Big); \qquad {\bf F}^d:={d}\Bigg(\big({\bf 1}-{\pmb{\mathcal K}^{\alpha}}\big)^{-1}\Big({\pmb{\mathcal K}^{\alpha}}(A^d)\Big)\Bigg).
\end{equation}
\noindent
Substitution in (\ref{Maxwell_1})-(\ref{Maxwell_2}) yields Maxwell's equations on GWIG:
\begin{numcases}{}
d{\bf \widehat{F}}=d{\bf F}^d, \label{Ext_Maxwell_1} \\
\delta{\bf \widehat{F}}=\delta{\bf F}^d.\label{Ext_Maxwell_2}
\end{numcases}
Some source terms are appeared in the right hand side due to the presence of the primary field $A^d\neq 0$ and non-zero dilations $\pmb{\chi}\neq 0$. 

In the sequel, we suggest a general solution for system (\ref{Ext_Maxwell_1}--\ref{Ext_Maxwell_2}). Moreover, the conservation, hyperbolic,  and elliptic equations on GWIG are obtained from this system in Appendix \ref{append1}. Furthermore, we shall discuss the solutions of the elliptic equation in details.     
\subsection{A general solution for system (\ref{Ext_Maxwell_1}--\ref{Ext_Maxwell_2})}\label{sec_4.1}
\noindent
The internal symmetry of system (\ref{Ext_Maxwell_1}--\ref{Ext_Maxwell_2}) might be enough as a proof for the following statement:
\begin{theorem}\label{theorem1}
For the given smooth $A^d$ and ${\pmb{\mathcal K}^{\alpha}}$, if $\widehat{A}$ solve system (\ref{Ext_Maxwell_1}--\ref{Ext_Maxwell_2}), then it can be written locally as
\begin{equation}\label{theorem-eq-1}
\widehat{A}=({\bf 1}-{\pmb{\mathcal K}}^{\alpha})(A)+{\pmb{\mathcal K}^{\alpha}}A^d,
\end{equation}
where $A$ solves system (\ref{Maxwell_1}--\ref{Maxwell_2}). \QEDA
\end{theorem}
\noindent
Now, note that $\lim_{\kappa \to 1}\widehat{A}=A^d$. This is the way that we impose boundary conditions on system (\ref{Maxwell_1}--\ref{Maxwell_2}). Depending on the distribution of $A^d$, it can be Dirichlet or Neumann boundary condition, but here we merely study the Dirichlet one. Moreover, we do not study the dynamics of $\kappa$, that is, we assume that there is a given stationary $\kappa\neq \kappa(x^0)$, and we assume that ${\bf F}^d$ is smooth enough. It is comparable with Penrose conformal infinity, as we let $\lambda\to \infty$ (equivalently $\kappa\to 1$). However, here this is the dual frame that goes to zero, and the initial Riemannian volume remains unchanged (see item \ref{item4-3} of $\S$ \ref{R-W}).\\ 
Consider homogeneous Maxwell's equations as an initial/boundary value problem:
\begin{numcases}{}
d{\bf F}=0, \qquad x^\mu\neq 0, \label{IBV_1} \\
\delta{\bf F}=0, \qquad x^\mu\neq 0, \label{IBV_2}\\
{\bf F}={\bf F}^d,  \qquad x^\mu={0}, \label{IBV_3}\\
{\bf F}={\bf F}^d, \qquad x^0>0,x^j={0}, \label{IBV_4}
\end{numcases}
where $j=1,...,3$ are the spatial coordinate indices. We assume that the initial data are smooth enough. The equations are defined on $x={\Bbb R}^4\setminus 0$, and the system has Dirichlet initial/boundary conditions at $x^\mu=0$. 

\noindent
Now, instead of solution of the above system, we suggest solution of system: 
\begin{numcases}{}
d{\bf \widehat{F}}_a=d{\bf F}^d, \label{Ext_IBV_1} \\
\delta{\bf \widehat{F}}_a=\delta{\bf F}^d,\label{Ext_IBV_2}
\end{numcases}
on ${\Bbb R}^4$, in which $\pmb{\mathcal{K}}^\alpha_a={\rm diag}(1-e^{-{\mathsf w}_j\delta_a})$ is a regularized $\pmb{\mathcal K}^\alpha$ that:
\begin{numcases}{}
\pmb{\mathcal K}^{\alpha}({\Bbb R}^4\setminus 0)=\bf 0, \label{non_reg_K_1} \\
\pmb{\mathcal K}^{\alpha}(0)=\bf 1,\label{non_reg_K_2}
\end{numcases}
and $\lambda(x)=:\lim_{a\to 0}\lambda_a(x)=:\lim_{a\to 0}\delta_a(x)$ is substituted in Eq. (\ref{kap_lamb_rel}), in which $\delta_a(x)$ is the four dimensional regularized Dirac delta function, regularized by the radius $a$ of a four dimensional ball placed on $x=0$.\\ 
Now, if $\widehat{A}_a$ solves system (\ref{Ext_IBV_1})--(\ref{Ext_IBV_2}), then theorem \ref{theorem1} guarantees that ${\bf F}=\lim_{a\to 0}\widehat{\bf F}_a$, where: 
\begin{equation}
{\bf \widehat{F}}_a:={d}\Big(\big({\bf 1}-{\pmb{\mathcal K}^{\alpha}_a}\big)^{-1}\big(\widehat{A}_a\big)\Big).
\end{equation}
\noindent
As an example of the above procedure, we shall use it in finding the solution of the (singularity-free) Laplace equation (\ref{mod_Laplace_1}) on GWIG. 
\section{A singularity-free potential theory}\label{gen_potential_theory}
\noindent
In Eq. (\ref{mod_Laplace_1}), we assume ${\mathsf z}_j=1$, ${\mathsf w}_j={\mathsf w}$, and we define:
\begin{equation}\label{kappa_2}
\tilde{\kappa}(x):=1-(1-\kappa)^{\mathsf w}=1-e^{-{\mathsf w}\lambda(x)}=\frac{e^{\mathsf{w}\lambda(x)}-1}{e^{\mathsf{w}\lambda(x)}},
\end{equation}  
for the convinience. Then, recalling $\nabla(e^{\mathsf{w}\lambda})/e^{\mathsf{w}\lambda}=\mathsf{w}\nabla \lambda$, one can write Eq. (\ref{mod_Laplace_1}) as:
\begin{equation}
\widehat{\Delta}\widehat{\phi}=\widehat{\Delta}(\tilde{\kappa}\phi^d), \label{Mod_Lap_1}
\end{equation}
where 
\begin{equation}\label{Mod_Lap_2}
\widehat{\Delta}:=\Delta+2\mathsf{w}\nabla\lambda\cdot \nabla+\mathsf{w}\left(\Delta\lambda+\mathsf{w}(\nabla\lambda)^2\right).
\end{equation}
Now, Theorem \ref{theorem1} is applicable directly:
\begin{corol}\label{corol_1}
For the given smooth $\phi^d$ and $\tilde{\kappa}$, if $\widehat{\phi}$ solves equation (\ref{Mod_Lap_1}), then it can be written locally as
\begin{equation}\label{corr-eq-1}
\widehat{\phi}=(1-\tilde{\kappa})\phi^h+\tilde{\kappa}\phi^d,
\end{equation}
where $\phi^h=\phi_{\rm Riemannian}$ is harmonic, that is, $\Delta\phi^h=0$.
\end{corol}
As an example, we use Eq. (\ref{Mod_Lap_1}) and corollary \ref{corol_1} to address a classical physics problem, that is, obtaining a non-singular and stable model for a finite size (non-point) charged particle.
\section{A non-singular model for a finite-size fixed charged particle}\label{non-sing}
The longstanding difficulty of defining a non-singular point source in the classical potential theory can be solved in the above generalized Weyl potential theory (introduced by Eq. (\ref{Mod_Lap_1})). This is mainly because infinite dilation removes zero volumes (as item \ref{item4-3} in $\S$ \ref{R-W} emphasizes). 
 
It is aimed to find a model for a fixed (non-moving) charged particle, with an electric charge $Q$ and a radius $a\to 0$, placed at the origin of a spherical coordinate. By definition of a regularized three dimensional Dirac delta function $\delta({\bf r})=:\lim_{a\to 0}\delta_a({\bf r})$, we denote  the regularized $\tilde{\kappa}$ as $\kappa_a(\bf{r})$ (see Eq. (\ref{kappa_2})). Then, by assuming $\mathsf{w}=2$, and substitution of  $\lambda({\bf r}):=\delta_{a}({\bf r})$ in Eq. (\ref{Mod_Lap_1}), one obtains:
\begin{equation}\label{Eq18}
\widehat{\Delta}\widehat{\phi}=\widehat{\Delta}(\kappa_{a} \phi^d),
\end{equation}
where
\begin{equation}\label{Eq19}
\widehat{\Delta}:=\Delta +4\nabla\delta_{a}\cdot\nabla+2\left(\Delta\delta_{a}+2(\nabla\delta_{a})^2\right).
\end{equation}
\noindent
The solution of Eq. (\ref{Eq18}) is our model for the charged particle.
\begin{Note}
From Eq. (\ref{kappa_2}), one can see that $\kappa_a(r)$ is a regularized indicator function:
\begin{numcases}
{\lim_{a\to 0}\kappa_a({\bf r})=\lim_{a\to 0}\Big(\frac{e^{\mathsf{w}\delta_a({\bf r})}-1}{e^{\mathsf{w}\delta_a({\bf r})}}\Big)={\bf 1}(\bf{r})=}\label{indicator_func}
1,  ~~for~ {\bf r}=0 , \\
0,   ~~for~ {\bf r}\neq 0, 
\end{numcases} 
It means that ${\rm supp}\left( \widehat{\Delta}(\kappa_{a} \phi^d) \right) $ (the R.H.S of Eq. (\ref{Eq18})) is the neighborhood of ${\bf r}=0$. 
\end{Note}
\begin{rem}
In the classical potential theory, by $a\to 0$ the Riemannian volume of the particle goes to zero, in contrast to here that it remains unchanged (see item \ref{item4-3} in $\S$ \ref{R-W}). Therefore, the singularity is removed. It also removes the classical difficulty of infinite energy of the vacuum, if we add the condition $\lambda\neq\infty$ to the definition of the (classical) vacuum state. 
\end{rem}
The solution can be obtained from corollary \ref{corol_1}. It is just needed to find $\phi^h$ and $\phi^d$. We shall find them in non-dimensionalized form, denoting by $(\breve{\cdot})$ symbol. 
\subsection{The harmonic solution $\phi^h$}
This is the solution outside the particle, that is, on Riemannian geometry. From the classical potential theory:
\begin{equation}\label{harmonic_1}
\breve{\phi}^h(\breve{r})=\frac{1}{\breve{r}},
\end{equation}
where $\breve{r}:=r/a$, $\breve{\phi}^h:=\phi_h/\phi_a$, and $\phi_a:=\phi(a)=Q/(4\pi\epsilon_0 a)$, and $\epsilon_0$ is the vacuum permittivity. 
\subsection{The dark field solution $\phi^d$}
Inside the particle, where $\kappa_a\to 1$, we have $\phi\to\phi^d$, regardless of the particular distribution of $\phi^d$. Therefore, one can construct locally a suitable non-singular $\phi^d$. Here we consider merely the case of a constant $\phi^d$, resulting in imposition of Dirichlet boundary condition on the classical Laplace equation \footnote{The type of boundary condition being imposed is dependent on the particular distribution of $\phi^d$. With this regards, the case of a constant right hand side of Eq. (\ref{Eq18}), that is, $\widehat{\Delta}(\kappa_{a} \phi^d)=q_d\neq q^d(x)$, might impose the Neumann boundary condition on the classical Laplace equation. We do not treat it in the present article.}. 

Recalling the gauge freedom of the potential theory, if $\phi^h$ is a solution, then $\phi^h+\phi^\prime$ is also a solution for any arbitrary $\phi^\prime$. Then, by considering Eq. (\ref{corr-eq-1}), one can conclude that, $\phi^d$ is arbitrary. Here, without loss of generality, we choose:
\begin{equation}\label{dark_1}
\phi^d=\phi_a.
\end{equation}
\begin{figure}
	\centering
		\includegraphics[scale=.05]{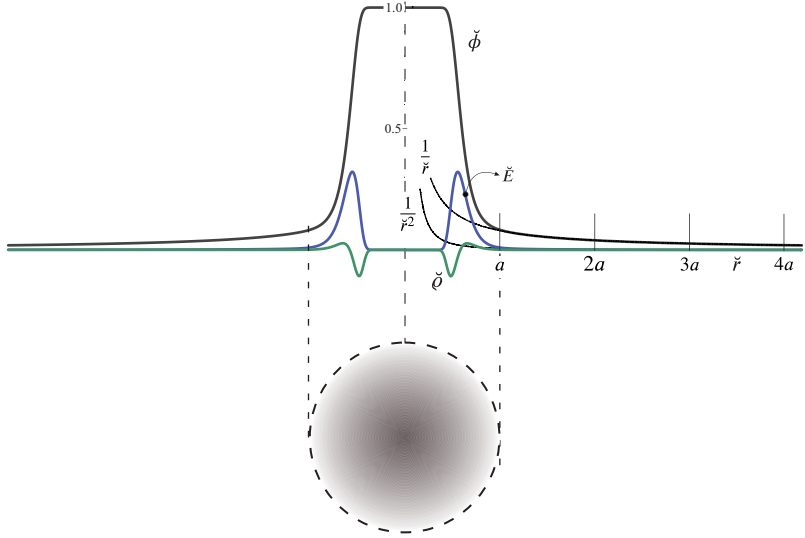}
	\caption{The values of $\breve{\phi}$, $\breve{E}$, and $\breve{\varrho}$ are shown versus $\breve{r}$ inside and outside of the particle. Moreover, $\breve{\phi}$ and $\breve{E}$ are compared with the classical $\breve{r}^{-1}$ and $\breve{r}^{-2}$.}
	\label{FIG_1}
\end{figure}
\subsection{The model}
Suitable $\phi^h$ and $\phi^d$ are found. Therefore, by substitution of (\ref{harmonic_1}) and (\ref{dark_1}) in (\ref{corr-eq-1}), and definition of $\breve{\phi}:=\widehat{\phi}/\phi_a$, one obtains:
\begin{equation}\label{particle_model_1}
\breve{\phi}(\breve{r}):=\frac{\widehat{\phi}(r)}{\phi_a}=\frac{1}{\breve{r}}+\big(1-\frac{1}{\breve{r}}\big)\kappa_{a}(\breve{r}). 
\end{equation}
This is the solution of Eq. (\ref{Eq18}), which is our non-singular model for a finite size charged particle with a constant $\phi^d$. It can be checked easily that it is smooth on ${\Bbb R}^3$ for any $a\to 0$.\\
\noindent
Before investigation of the model, we obtain the following quantities:
\begin{enumerate}
\item Eqs. (\ref{mod_Laplace_1}) and (\ref{Eq18}), are obtained for the components of $\eta^{-1}(\widehat{A})$ on $TM$. As a result, $\widehat{\phi}$ is defined on $(M,g)$. Therefore, we define the electric field $\widehat{\bf E}\in\Gamma(TM)$ as $\widehat{\bf E}:=-(d\widehat{\phi})^\sharp$, and its non-dimensional form: 
\begin{equation}\label{E_eq}
{\bf{\breve{E}}}(\breve{r}):=\frac{{\bf E}(r)}{{E}(a)}=-\frac{\partial\breve{\phi}}{\partial{\breve{r}}}=\Big(\frac{1-\kappa_a}{\breve{r}^2}-\frac{\breve{r}-1}{\breve{r}}\kappa_a^\prime\Big)\partial_{\breve{r}},
\end{equation}
where ${E}(a):=Q/(4\pi\epsilon_0 a^2)$ and $\kappa_a^\prime:=\partial_{\breve{r}}\kappa_a$ is the derivative of $\kappa_a$ with respect to $\breve{r}$. It should be noted that for $\kappa_a=\kappa_a^\prime=0$, we have $\widehat{\bf E}={\bf E}=(d\phi)^\sharp$.  
\item By definition of the electric charge density $q:=\widehat{\varrho} {d{\mathcal V}}$, where $\epsilon_0 ^{-1}\widehat{\varrho}=\delta d (\widehat{\phi})=div(\widehat{\bf E})$, one can write:
\begin{equation}\label{chrge_dens}
\breve{\varrho}(\breve{r}):=\frac{\widehat{\varrho}(r)}{\varrho_0}=\left(\frac{2}{3\breve{r}}\right)\kappa_a^\prime+\frac{1}{3}\left(1-\frac{1}{\breve{r}} \right)\kappa_a^{\prime\prime},
\end{equation} 
where $\varrho_0:=Q/(4/3\pi a^3)$ is the uniform electric charge density, and $\kappa_a^{\prime\prime}:=\partial^2_{\breve{r}}\kappa_a$.  
\end{enumerate}
In Fig. \ref{FIG_1} the quantities ${\phi}$ and $\tilde{E}:={\bf{\tilde{E}}}\cdot{\bf{\hat{r}}}$ and $\tilde{q}$ are shown  versus $\tilde{r}$. The figure is obtained for the particular regularized Dirac delta function:
\begin{equation}\label{delt_p}
\delta_a(\tilde{r}):=\frac{\beta}{a}\frac{e^{-\tilde{r}}}{(1+e^{-\tilde{r}})^2},
\end{equation}
where the adjustment constant $\beta\approx 10$ is chosen such that ${\rm supp}(\kappa_a(\tilde{r}))\sim a$.\\

\noindent
The following points are noticeable about the model:
\begin{enumerate}
\item The model is not completely scale-free.\\ 
For a fixed Riemannian radius $a$, the particle and its surroundings (Riemannian and Weylian as a whole), may be expanded or be contracted freely; because $\kappa>0$ is defined relative to the Riemannian geometry $\kappa=0$. But, for a Riemannian observer, the model depends on $a$.
\item It is easy to show that $\phi^d$ is arbitrary. In fact, for a null dark (primary) field $\phi^d=0$, one would get
\begin{equation}\label{null_1}
\breve{\phi}(\breve{r})=\frac{1-\kappa_{a}}{\breve{r}},
\end{equation}
instead of Eq. (\ref{particle_model_1}), sharing main properties with it.
\item Eq. (\ref{chrge_dens}) may be seen as a geometric interpretation of the electric charges. As it shows: 
\begin{quote}
``{\it The electric charge is a result of change of dilation of the electric field.}''
\end{quote}
\end{enumerate}
In general, consider the Laplace equation with the Dirichlet boundary condition:
\begin{numcases}{}\label{Laplace_3}
\Delta \varphi=0, \qquad x\in{\Bbb R}^3\setminus 0; \label{Laplace_3_1} \\
\varphi(0)=\varphi^d. \label{Laplace_3_2}
\end{numcases}
As a special case of the solution provided in $\S$ \ref{sec_4.1}, we suggest solution of Eq. (\ref{Eq18}) instead of system (\ref{Laplace_3_1})-(\ref{Laplace_3_2}), which according to corollary \ref{corol_1} has the solution $\varphi=\lim_{a\to 0} \varphi_{\delta_a}$, where: 
\begin{equation}\label{lap_sol}
\varphi_{\delta_a}=(1-\kappa_a)\frac{\varphi^d}{r^2}+\kappa_a\varphi^d.
\end{equation}
The superiority of the above solution, in comparison to the classical fundamental solutions, is obvious; it remains smooth by $a\to 0$.
\section{Conclusions}\label{conclusion}
By affine transformations of the pair ``vector space--dual space'', both dilation and arbitrariness of the physical gauges, are included in the geometric vector spaces. Application of these transformations on the pair tangent--cotangent bundles of a (pseudo)-Riemannian manifold results in a generalized Weyl integrable geometry (GWIG), containing some primary fields which can be interpreted as the dark fields. In this framework, the classical Weyl integrable geometry (CWIG) has already included interactions with null dark fields.

The GWIG introduces a new internal symmetry in Maxwell's equations by which the interactions with the primary (dark) fields are explainable. The gauge-free conservation, hyperbolic and elliptic equations are derived on GWIG from Maxwell's equations.

In GWIG, an infinitely dilated manifold consists of isolated points,  each one has a pair tangent ``vector--covector'' of the primary fields, instead of the tangent ``vector space--dual space''. Based on this property, a method, comparable with Penrose conformal infinity, is suggested that imposes canonically the boundary conditions on Maxwell's equations and its sub-sets. A singularity-free potential theory is constructed from the gauge-free elliptic equation on GWIG. The theory is singularity--free because by approaching dilation to infinity, the values remain defined. The potential theory is used in the construction of a non-singular model for a classical (non-quantum) fixed (non-moving) finite-size charged particle. It solves the old difficulty of infinite energy of the classical vacuum.    
\appendix
\section{Conservation, parabolic and elliptic equations on the GWIG}
\label{append1}
To be comparable with their classical counterparts, we obtain these equations with respect to a Riemannian observer $\mathscr R$ (see $\S$ \ref{R-W}). We will use both $\kappa$ and $\lambda$ for the convenience, but they are related via Eq. (\ref{kap_lamb_rel}).

On the spacetime manifold $({\Bbb R}^4,\eta)$, where $\eta:={\rm diag}(1,-1,-1,-1)$, we consider a dilation tensor $\pmb{\chi}:={\rm diag}(e^{-{\mathsf z}_j\lambda})$ inducing Weyl weights $\pmb{\mathsf w}=\pmb{\mathsf w}(\pmb{\mathsf z})$. The morphisms $(\pmb{\mathfrak F}_X,\pmb{\mathfrak F}_V)$ are then defined on the manifold, resulting in the metric $\widehat{\eta}_{\mathscr R}=\chi^{-1}(\eta)={\rm diag}(e^{{\mathsf z}_0\lambda},-e^{{\mathsf z}_1\lambda},-e^{{\mathsf z}_2\lambda},-e^{{\mathsf z}_3\lambda})$. Moreover, assume that there is a dark (primary) 1-form $\alpha^d=(v_d)^{\flat}$ on the manifold. 
\begin{rem}\label{Rem5}
In what follows, we need the four-gradient with respect to $\mathscr R$. According to corollary \ref{coroll_1}, it reads: $\partial^\mu=\widehat{\eta}^{\mu\nu}\partial_\nu$. Moreover, note that $\eta$ is diagonal. Therefore, according to corollary \ref{corol1}, $\widehat{A}^\mu=\eta^{\mu\nu}\widehat{A}_\nu$.
\end{rem}
The procedure is exactly the same as the procedure of obtaining the classical conservation, hyperbolic and elliptic equations from the classical Maxwell's equations, that is: the gauge fixing $\delta A\equiv 0$ sets $d\delta A=0$. Then by satisfying $\delta {\bf F}=0$, we set $\delta d A=0$ indeed. Together, they result in: 
\begin{equation}\label{derham1}
{\Delta}_H A:=(d\delta+\delta d)A={\Delta}_H\Big(\big({\bf 1}-{\pmb{\mathcal K}}\big)^{-1}\big(\widehat{A}-{\pmb{\mathcal K}}(A^d)\big)\Big)=0, 
\end{equation}
where $\Delta_H$ is the Hodge-de Rham Laplacian.

However, depending on whether the fields are non-stationary $\partial_t:=\partial_0\neq 0$ or stationary $\partial_t=0$, we obtain different types of equations.
\subsection{The non-stationary fields}
For the non-stationary fields $\partial_t\neq 0$, the gauge fixing ({\it i.e.}, Lorenz condition) $\delta A=0$ reads:
\begin{equation}\label{lor_1}
-\frac{\partial}{\partial t}\left(e^{{\mathsf w}_0\lambda}\big(\widehat{A}_0-{\mathcal K}^\alpha_0A^d_0\big)\right)+\frac{\partial}{\partial x^j}\left(e^{{\mathsf w}_j\lambda}\big(\widehat{A}_j-{\mathcal K}^\alpha_jA^d_j\big) \right)=0.
\end{equation}
Then, by writing $\partial^\mu=\widehat{\eta}^{\mu\nu}{\partial}_\nu$, $\widehat{A}^\mu=\eta^{\mu\nu}\widehat{A}_\nu$ (see remark \ref{Rem5}), and $({A}^d)^\mu=\eta^{\mu\nu}(A^d)_\nu$, the divergence of $\widehat{A}^\mu$  may be written as: 
\begin{equation}\label{Ext_Lor}
e^{-{\mathsf z}_0\lambda}\partial_0\left(e^{{\mathsf w}_0\lambda}\widehat{A}^0\right)+e^{-{\mathsf z}_j\lambda}\partial_j\left(e^{{\mathsf w}_j\lambda}\widehat{A}^j \right)={\mathcal R}_1,
\end{equation}
where: 
\begin{equation}
{\mathcal R}_1:=e^{-{\mathsf z}_\mu\lambda}\partial_\mu\Big(e^{\mathsf w_\mu\lambda}(1-e^{-\mathsf w_\mu\lambda})(A^d)^\mu \Big).
\end{equation}
Equation (\ref{Ext_Lor}) is the conservation equation on GWIG, and its solution yields Lorenz gauge $\widehat{A}_\mu$ on GWIG. Substitution of $\widehat{A}$ in Maxwell's equations results in four decoupled equations with a generic form:
\begin{equation}\label{mod_wave_1}
-\partial^2_{0} \Big(e^{\mathsf w_\mu\lambda}\big(\widehat{\phi}-{\mathcal K}_\mu\phi^d\big)\Big)+{\partial^2_{j}}\Big(e^{{\mathsf w}_\mu\lambda}\big(\widehat{\phi}-{\mathcal K}_\mu\phi^d\big)\Big)=0,    
\end{equation}
where $\widehat{\phi}$ and ${\phi}^d$ stand for any $\mu^{\rm th}$ component of $\widehat{A}$, and $A^d$. This equation is equivalent with the equation $\partial^\nu\partial_\nu(\eta^{\mu\sigma}{A}_\sigma)=0$ on $\widehat{T}M$, which is the wave equation on GWIG:
\begin{equation}\label{mod_wave_2}
\Big(e^{-{\mathsf z}_0\lambda}\partial^2_{0}-e^{-{\mathsf z}_j\lambda}\partial^2_{j}\Big)(e^{{\mathsf w}_\mu\lambda}\widehat{\phi})={\mathcal R}_2, 
\end{equation}
where
\begin{equation}\label{mod_wave_3}
{\mathcal R}_2:=\Big(e^{-{\mathsf z}_0\lambda}\partial^2_{0}-e^{-{\mathsf z}_j\lambda}\partial^2_{j}\Big)\big(e^{\mathsf w_\mu\lambda}(1-e^{{-\mathsf w}_\mu\lambda})\widehat{\phi}^d\big).
\end{equation}
We emphasis that the behavior of equation may be different because ${\mathsf w}_\mu$ may be different from a component to the other.
\subsection{The stationary fields}
For $\partial_t=0$ where the the geometry is elliptic, the gauge fixing ${\delta} A=0$ results in Coulomb gauge $\widehat{A}$ on GWIG  satisfying:
\begin{equation}\label{Columb_1}
e^{-{\mathsf z}_j\lambda}\partial_j\Big(e^{{\mathsf w}_j\lambda}\widehat{A}^j\Big)=e^{-{\mathsf z}_j\lambda}\partial_j\Big(e^{\mathsf w_j\lambda}(1-e^{-{\mathsf w}_j\lambda})({A}^d)^j\Big), 
\end{equation}      
on $\widehat{T}M$. Substitution in Maxwell's equations results in four decoupled Laplace equations on GWIG with a generic form:
\begin{equation}\label{mod_Laplace_1}
e^{-{\mathsf z}_j\lambda}\partial^2_{j}\big(e^{{\mathsf w}_\mu\lambda}\widehat{\phi}\big)=e^{-{\mathsf z}_j\lambda}\partial^2_{j}\big(e^{\mathsf w_\mu\lambda}(1-e^{-{\mathsf w}_\mu\lambda}){\phi}^d\big).
\end{equation}
where $\widehat{\phi}$ and $\phi^d$ stand for any $\mu^{\rm th}$ component of $\eta^{-1}(\widehat{A})$ and $\eta^{-1}({A}^d)$. This equation results in a potential theory on GWIG, as it is discussed in $\S$ \ref{gen_potential_theory}.

\end{document}